\documentclass[conference,a4paper]{IEEEtran}
\IEEEoverridecommandlockouts

\usepackage{amssymb, cite}
\usepackage{bm,bbm}
\usepackage{mathtools}
\usepackage{multirow,color,amsfonts}
\usepackage{tabulary}
\usepackage{subfigure}
\usepackage{graphicx}
\usepackage{setspace}
\usepackage{enumerate}
\usepackage{algorithm,algpseudocode}
\usepackage{comment}

\usepackage[utf8]{inputenc} 
\usepackage[T1]{fontenc}
\usepackage{url}
\usepackage{ifthen}
\usepackage{cite}
\usepackage{mathtools}
\usepackage{enumitem}

\DeclarePairedDelimiterX{\infdivx}[2]{(}{)}{
  #1\;\delimsize\|\;#2
}

\usepackage{amsthm}
\usepackage{amsmath}

\newtheorem{theo}{Theorem}
\newtheorem{defi}{Definition}

\newtheorem{prop}{Proposition}

\newtheorem{ex}{Example}

\newcommand{\FCGE}{{fractional graph entropy}}

\newcommand{\FCN}{{fractional chromatic number}}

\newcommand{\FCE}{{fractional chromatic entropy}}

\newcommand{\HGfrac}{{H^{f}_{G_{X_1}}(X_1\vert X_2)}}

\newcommand{\HGchifrac}{{H^{\chi_f}_{G_{X_1}}(X_1\vert X_2)}}

\newcommand{\GXone}{G_{X_1}^w}
\newcommand{\GXtwo}{G_{X_2}^w}

\newcommand{\GXonen}{G_{{\bf X}_1}^{n,w}}
\newcommand{\GXtwon}{G_{{\bf X}_2}^{n,w}}

\newcommand{\GXonetwo}{G_{{\bf X}_1}^{2,w}}

\newcommand{\Graph}{G_{X_1}}

\newcommand{\nPowerGraph}{G^n_{{\bf X_1}}}

\newcommand{\HGchipowernfrac}{{H^{\chi_f}_{G^n_{{\bf X_1}}}({\bf X_1}\vert {\bf X_2})}}

\newcommand{\coloringf}{c^f_{{G_{X_1}}}(X_1)}
\newcommand{\coloringxf}{c^f_{{G_{X_1}}}}
\newcommand{\coloringpowernf}{c^f_{{G^n_{{\bf X_1}}}}({\bf X_1})}

\newcommand{\coloringpowernxf}{c^f_{{G^n_{{\bf X_1}}}}}

\begin{document}
\title{Weighted Graph Coloring for Quantized Computing}

\author{
  \IEEEauthorblockN{Derya Malak}
  \IEEEauthorblockA{Communication Systems Department,   EURECOM, Sophia Antipolis, 06904  France\\
derya.malak@eurecom.fr
}
\thanks{Funded by the European Union (ERC, SENSIBILITÉ, 101077361). Views and opinions expressed are however those of the author only and do not necessarily reflect those of the European Union or the European Research Council. Neither the European Union nor the granting authority can be held responsible for them.}
}

\maketitle

\begin{abstract}
We consider the problem of distributed lossless computation of a function of two sources by one common user. To do so, we first build a bipartite graph, where two disjoint parts denote the individual source outcomes. We then project the bipartite graph onto each source to obtain an edge-weighted characteristic graph (EWCG), where edge weights capture the function's structure, by how much the source outcomes are to be distinguished, generalizing the classical notion of characteristic graphs. Via exploiting the notions of characteristic graphs, the fractional coloring of such graphs, and edge weights, the sources separately build multi-fold graphs that capture vector-valued source sequences, determine vertex colorings for such graphs, encode these colorings, and send them to the user that performs minimum-entropy decoding on its received information to recover the desired function in an asymptotically lossless manner. For the proposed EWCG compression setup, we characterize the fundamental limits of distributed compression, verify the communication complexity through an example, contrast it with traditional coloring schemes, and demonstrate that we can attain compression gains higher than $\% 30$ over traditional coloring.
\end{abstract}

\section{Introduction}
\label{sec:intro}

Over the past years, we have been experiencing an ever-increasing demand for computationally-intensive tasks, motivating us to devise new parallel processing techniques to speed up and efficiently distribute computations across groups of servers.  
In modern distributed computing, a primary concern is communication cost. While parallel processing to distribute communication can reduce the need for coordination and alleviate this cost, reduction of the same communication cost is challenged due to issues of scalability \cite{soleymani2021analog}, accuracy \cite{wang2021price}, low capacity edges \cite{ho2006random}, and stragglers \cite{behrouzi2020efficient} in distributed computing.

\subsection{Related Work} 
\label{sec:related_work}

{\bf Distributed coded computation.} There have been various efforts to mitigate the communication cost in distributed computing following Yao's seminal work in \cite{yao1979some} on communication complexity. Some recent breakthroughs in this direction include coded computing 
\cite{soleymani2021analog,YuRavSoAve2018,reisizadeh2019coded,prakash2020coded}, and distributed computation of, e.g., matrix products~\cite{jia2021capacity,wan2022cache,chang2018capacity,chen2021gcsa}, distributed batch computation \cite{jia2021cross}, 
matrix multiplication with stragglers \cite{li2021flexible}, secure matrix multiplication \cite{jia2021capacity}, \cite{chang2018capacity}, \cite{chen2021gcsa}, 
cache-aided general linear function retrieval \cite{wan2020cache}, 
and linearly separable functions \cite{khalesi2022multi,wan2021distributed,wan2022secure}.

{\bf Distributed source and functional compression.} Other attempts have been inspired from the seminal work of Slepian-Wolf \cite{SlepWolf1973} on distributed source compression, the rate-distortion coding models of Wyner-Ziv with side information  \cite{WynZiv1976}, and for lossy source coding \cite{yamamoto1982wyner}, toward function computation. These works include~\cite{OR01,Kor73,AO96,feizi2014network} that consider function computation over networks, as well as \cite{feizi2014network} and \cite{doshi2007source}, considering the generalization to functional rate-distortion, and \cite{basu2022hypergraph} and \cite{basu2020functional}, focusing on hypergraph-based source coding and function approximation under maximal distortion.  
Recent works also include hyperbinning for distributed function quantization \cite{Malak2022hyperbin},  generalizing the orthogonal binning ideas in Slepian-Wolf coding \cite{cover1975proof},  and fractional coloring-based distributed computation \cite{malak2022fractional} that reduces complexity of \cite{Kor73}.

{\bf Coding for specific functions and channels.} The communication cost is also affected by the nature of the computed function. 
Examples include K{\"o}rner-Marton' encoding problem for computing {\em modulo-two sum of binary sources} \cite{korner1979encode}, the generalization of K{\"o}rner-Marton's problem to a two-terminal source coding scheme with {\em common sum reconstruction} \cite{adikari2022two}, which has applications in distributed stochastic gradient descent, power iteration, and Max-Lloyd's algorithms \cite{suresh2017distributed} to compute large-scale averages over a large number of servers.  
Han and Kobayashi have established necessary and sufficient conditions on functions such that the Slepian-Wolf region is optimal for distributed lossless computing \cite{han1987dichotomy}. 
The authors in \cite{witsenhausen1976zero,ahlswede1979coloring,korner1998zero} have explored the {\em combinatorial aspects of zero-error source coding} to compress correlated sources separately or for compression 
with decoder side information. 
The {\em joint source-channel scheme} of Cover, El Gamal, and Salehi uses the source correlations to achieve a collaborative gain and create channel input distributions adapted to the channel \cite{cover1980multiple}.  
To that end, Nazer and Gastpar have devised designs for distributed computing over multiple access channels  \cite{nazer2007computation}, and {\em structured coding} for Gaussian networks \cite{nazer2007lattice}. 
Distributed computing of {\em functions of structured sources} has been studied in \cite{malak2023Structured} while benefiting from low-rate side information provided by a helper node. Focusing on the case where the sources are jointly distributed according to a very general mixture model, an achievable coding scheme has been provided to substantially reduce the communication cost of distributed computing by exploiting the nature of the joint distribution of the sources, the side information, as well as the symmetry enjoyed by the desired functions.

\subsection{Overview and Contributions}
\label{sec:contributions}

We focus on distributed computing of a function of two jointly distributed finite alphabet sources at a user.  
We pose this problem as an edge-weighted characteristic graph (EWCG) compression problem. To do so, we build a bipartite graph\footnote{Important classes of bipartite networks are the collaboration network 
and the opinion network. 
They are significant in information and economic systems, social networks, opinion networks and recommendation systems \cite{zhou2007bipartite}.} 
where two disjoint parts denote the individual source outcomes, and the edges capture the joint source distribution.  

Our main contributions can be summarized as follows:

\begin{itemize}[leftmargin=*]
\item {\bf \emph{Edge-weighted $b$-fold compression.}} 
We propose an EWCG encoding scheme to provide low-complexity compression for computing, where we describe the weights by the joint source distribution and the function. 
An EWCG is a fractionally colored characteristic graph built by each source as an edge-weighted projection of the bipartite graph (Sect.~\ref{sec:model_defns}). 
To capture the unequal edge weights, the source devises $b$ characteristic graphs (one characteristic graph per source coordinate, see App.~\ref{preliminary}), where the edge weights in EWCG are quantized across these graphs (which we will detail via Example~\ref{uniform_weighted_example}).

In an EWCG, a vertex captures a {\em $b$-fold, i.e., vector-valued, source value}, and is given $b$ colors out of $a$ available colors, where $b$ captures the quantization depth of each source. 
The edge weights are used to determine $a$, $b$, and the overlap of colorings for any vertex pair (Sect.~\ref{sec:results}), upon which each source establishes and encodes the vertex colorings of its EWCG.

\item {\bf \emph{Edge-weighted fractional chromatic entropy.}} 
The fractional chromatic number $\chi_f$ -- given by the limit in (\ref{fractional_graph_coloring}) in App.~\ref{preliminary} -- determines the communication complexity when the edges have unit weights. 
Using OR power graphs, we can exploit the gains in complexity through fractional coloring as the blocklength $n$ tends to infinity \cite[Ch. 3]{scheinerman2011fractional}. 
To that end, we generalize the definition of $\chi_f$ via EWCGs to provide a lower communication complexity (Sect.~\ref{sec:results}).

\item {\bf \emph{Joint quantization and distributed functional compression via  
EWCGs.}} In the edge-weighted fractional coloring of {\em vector-valued  
sources}, $b$ is the quantization depth. 
The encoding rates for EWCGs are lower versus traditional or fractional coloring of graphs 
because the higher the value of $b$ 
is, the more refined  
the weights in an EWCG are, enabling a lower rate of compression per source coordinate. 
We characterize in (\ref{distance_between_vertices}) the number of disjoint colors between two vertices of an EWCG. 
We provide in Theorem \ref{fractional_graph_entropy_edge_weighted_graph} (Sect.~\ref{sec:results}) the encoding rate for a $b$-fold fractional coloring of EWCGs.

\item {\bf \emph{Numerical experiments.}} Contrasting it with the existing techniques via an example, EWCG exhibits significant savings in communication complexity by taking into account the structures of the sources (via the Slepian-Wolf theorem \cite{SlepWolf1973}) and the function (via the edge weights).
\end{itemize}

\subsection{Notation} 
For a random variable $X$ with a finite alphabet $\mathcal{X}$, $P_X$ denotes its probability mass function (PMF). Similarly, for variables $X_1$ and $X_2$, $P_{X_1,X_2}$ denotes the joint PMF of finite alphabet $X_1$ and $X_2$. 
We denote the probability of an event $A$ by $\mathbb{P}(A)$. Let the entropy function of a PMF ${\bf p}$ be $h({\bf p})=-\sum\nolimits_i p_i\log p_i$ where the logarithm is in base $2$, $h(p)$ be the binary entropy function with parameter $p$, and $H(X)=\mathbb{E}[-\log P_{X}(X)]$ be the Shannon entropy of $X$ drawn from  $P_{X}$. We denote by ${\bf X}_1^n=X_{11},X_{12},\dots, X_{1n}\in \mathcal{X}_1^n$ the length $n$ sequence of $X_1$ sampled from an $n$-fold finite alphabet $\mathcal{X}_1^n$. 
We let $[N]=\{1,2,\dots,N\}$, $N\in \mathbb{Z}^+$.

\section{Model and Problem Statement}
\label{sec:model_defns}

We pose the problem of distributed computation of a bivariate 
function $f(X_1,\,X_2)$ of the two sources $X_1$ and $X_2$ as a \emph{compression problem for the edge-weighted projections of a bipartite graph model that captures $P_{X_1,X_2}$}. 
For this partially distributed setting, we will devise  
an encoding scheme for EWCGs and  quantify the sum rate for computing $f(X_1,\,X_2)$, by exploiting the notions of characteristic graphs and their entropy \cite{Kor73,OR01,AO96,feizi2014network} and the concept of bipartite graph projection. 
For $M$ sources and computing $f(X_1,\, X_2,\dots,X_M)$, we can exploit the notion of $M-$partite graphs, which is left as future work.

\subsection{Bipartite Graph Representation}
\label{sec:bipartite_graph}

We construct a bipartite graph representation $G_f=(\mathcal{X}_1,\mathcal{X}_2,E)$ to compute the function $f(X_1,\,X_2)$, whose partition has the parts $\mathcal{X}_1$ and $\mathcal{X}_2$, which correspond to the set of realizations of the sources $X_1$ and $X_2$, respectively, and $E$ denotes the set of edges of $G_f$. The bipartite graph $G_f$ is derived from the joint distribution $P_{X_1,\,X_2}$, and $E$ captures the correlation between $X_1$ and $X_2$. More specifically, $G_f$ has the following properties:
\begin{enumerate}[leftmargin=*]
    \item The set of vertices $\mathcal{X}_1$ and $\mathcal{X}_2$ that partition $G_f$ are disjoint and correspond to the set of source realizations, i.e., the alphabets $\mathcal{X}_1$ and $\mathcal{X}_2$, respectively.
    \item $G_f$ is a balanced bipartite graph with $|\mathcal{X}_1|=|\mathcal{X}_2|$, i.e., the two subsets of vertices have the same cardinality.
    \item There is an edge between vertices $u_k\in \mathcal{X}_1$ and $v_l\in \mathcal{X}_2$, i.e., $(u_k,v_l)\in E$, if and only if $\mathbb{P}(X_1= u_k , X_2= v_l) >0$.
    \item If $u_k \in \mathcal{X}_1$ and $v_l\in \mathcal{X}_2$ are connected, i.e., $(u_k,v_l)\in E$, and  
    $(v_l,u_k)\in E$, then the symmetry of the edges does not imply that both edges yield the same function outcome. 
\end{enumerate}

If $G_f$ is complete, it has $|\mathcal{X}_1|\cdot|\mathcal{X}_2|$ edges and the number of distinct function outcomes is determined by the structure of $f(X_1,\,X_2)$. 
On the other hand, if $G_f$ is not connected, it may have more than one bipartition \cite{chartrand2019chromatic}. In that case,  encoding of $f(X_1,\,X_2)$ is facilitated upon the extraction of the bipartition information.  
We note that the sources do not have 
the full knowledge of $E$, as determined by $P_{X_1,X_2}$, but only the weights jointly determined by $P_{X_1,X_2}$ and $f(X_1,\,X_2)$. 
We assume that the edge weights are available,  
and can be learned via feedback, the study of which is left as future work.

\subsection{Weighted Bipartite Graphs through Projections of $G_f$}
\label{sec:weighted_graph}

Source one $X_1$ observes a weighted projection of $G_f$ onto a graph -- the $X_1$ projection of $G_f$ -- denoted by $\GXone$, and similarly for source two. For the EWCG of source one, given by $\GXone$, the edge weight between $u_{k_1},\, u_{k_2}\in \mathcal{X}_1$  of $\GXone$, denoted by $w(u_{k_1},u_{k_2})$, is set to be the weighted number of common neighbors in $X_2$. 
Hence, the notion $\GXone$ generalizes the concept of the characteristic graph $G_{X_1}$ detailed in App.~\ref{preliminary}. 
In this paper, 
we determine $\{w(u_{k_1},u_{k_2}),\,u_{k_1}, u_{k_2}\in \mathcal{X}_1\}$ as 
\begin{align}
\label{weighting}
w(u_{k_1},u_{k_2})
=\sum\limits_{ \underset{\prod_{k\in\{k_1,k_2\}}P_{X_1,X_2}(u_{k},v_l)>0}{v_l\in \mathcal{X}_2:\,f(u_{k_1},v_l)\neq f(u_{k_2},v_l)}} \hspace{-1.5cm} P_{X_1,X_2}([u_{k_1},u_{k_2}],v_l)\ ,
\end{align}
where $P_{X_1,X_2}([u_{k_1},u_{k_2}],v_l)=\sum_{k\in\{k_1,k_2\}}P_{X_1,X_2}(u_{k},v_l)$. 
The idea is similar for determining $w(v_{l_1},v_{l_2})$ 
of $\GXtwo$.

Similarly, towards realizing the limits of compression,  
for the $n$-th power graph of the EWCG 
$\GXone$, namely $\GXonen$, can be determined using 
the edge weight between the vertices ${\bf u}^n_{i}, {\bf u}^n_{j}\in \mathcal{X}_1^n$ of $\GXonen$, which 
is given as 
\begin{align}
w({\bf u}^n_{i},{\bf u}^n_{j})=\sum\limits_{{\bf v}^n_l\in \mathcal{X}_2^n:\, \underset{\prod_{k\in\{i,j\}}P_{{\bf X}^n_1,{\bf X}^n_2}({\bf u}^n_{k},{\bf v}^n_l)>0}{f({\bf u}^n_{i},{\bf v}^n_l)\neq f({\bf u}^n_{j},{\bf v}^n_l)}} \hspace{-1.5cm} P_{{\bf X}^n_1,{\bf X}^n_2}([{\bf u}^n_{i},{\bf u}^n_{j}],{\bf v}^n_l)\ .\nonumber 
\end{align}

We can note that for the standard construction $G_{X_1}$ of $X_1$ \cite{shannon1956zero}, \cite{Kor73}, as detailed in App. \ref{preliminary}, the edge weights satisfy
\begin{multline}
w(u_{k_1},u_{k_2})=\\
1_{\big|\big\{v_l\in \mathcal{X}_2\,:\,\prod_{k\in\{k_1,k_2\}}P_{X_1,X_2}(u_{k},v_l)>0,\, f(u_{k_1},v_l)\neq f(u_{k_2},v_l)\big\}\big|>0}\ .\nonumber
\end{multline}

In distributed compression, exploiting the notion of jointly typical sequences, it is possible for the user to estimate the number of ${\bf X}_2^n$ sequences jointly typical with ${\bf X}_1^n$ given ${\bf X}_1^n$. Hence, as a simplification of this paper's model in (\ref{weighting}), while still generalizing $G_{X_1}$, the weight $w(u_{k_1},u_{k_2})$ for  $u_{k_1},\,u_{k_2}\in\mathcal{X}_1$ can be set as the number of common neighbors in $\mathcal{X}_2$:
\begin{multline}
\label{classical_weight_rule}
w(u_{k_1},u_{k_2})= \\
\sum\limits_{v_l\in \mathcal{X}_2} 1_{\prod_{k\in\{k_1,k_2\}}P_{X_1,X_2}(u_{k},v_l)>0,\,f(u_{k_1},v_l)\neq f(u_{k_2},v_l)}\ .
\end{multline}
The edge weights in (\ref{weighting}) affect the quantization of the source outcomes through a $b$-tuple of graphs, which we detail next.

\section{Main Results}
\label{sec:results}

In this section, we provide an achievable encoding and decoding approach for asymptotically lossless distributed computation of $f(X_1,\,X_2)$, which is based on projecting the bipartite graph $G_f$ onto EWCGs and compressing the EWCGs.

\subsection{Valid Colorings of Edge-Weighted Graphs}
\label{sec:validcoloring_weighted_graphs}

In traditional coloring of an unweighted graph $G_{X_1}$, we note that given a pair of vertices $u_{k_1},\, u_{k_2}\in \mathcal{X}_1$ such that $w(u_{k_1}, u_{k_2})=0$, it implies that the two vertices can have identical colors $c_{G_{X_1}}(u_{k_1})=c_{G_{X_1}}(u_{k_2})$. On the other hand, $w(u_{k_1}, u_{k_2})>0$ implies $c_{G_{X_1}}(u_{k_1})\neq c_{G_{X_1}}(u_{k_2})$.

In fractional coloring of EWCGs, prior to a valid coloring of vertices of $\GXone$ and $\GXtwo$, we normalize each weight in~(\ref{weighting}) by $\max\, \{w(u_{k_1}, u_{k_2}),\,u_{k_1}, u_{k_2}\in \mathcal{X}_1\}$, and similarly for $\{w(v_{l_1},v_{l_2}),\,v_{l_1}, v_{l_2}\in \mathcal{X}_2\}$ of $\GXtwo$. 

We next let $c^f_{\GXone}(u_{k_1})$, $u_{k_1}\in \mathcal{X}_1$ be a valid fractional coloring with a $b$-fold coloring, where $u_{k_1}$ is assigned $b$ colors out of $a$ available colors. Note that the distance between colors $c^f_{\GXone}(u_{k_1})$ and $c^f_{\GXone}(u_{k_2})$, i.e., ${\rm dist}(c^f_{\GXone}(u_{k_1}),c^f_{\GXone}(u_{k_2}))$, is an increasing function of $w(u_{k_1}, u_{k_2})$ \cite{toivonen2011compression}. 
To that end, we stretch Defns. \ref{valid_ab_coloring} and \ref{fractional_chromatic_number} in App. \ref{preliminary} of the standard $a:b$ coloring, and adopt the following model.  
As in traditional coloring, for a given $u_{k_1},\, u_{k_2}\in\mathcal{X}_1$, in the special case when $w(u_{k_1}, u_{k_2})=0$, then the $b$-fold colors  
$c^f_{\GXone}(u_{k_1})$ and $c^f_{\GXone}(u_{k_2})$ could be identical, i.e., ${\rm dist}(c^f_{\GXone}(u_{k_1}),c^f_{\GXone}(u_{k_2}))=0$. On the other hand, when $w(u_{k_1}, u_{k_2})=1$, then the $b$-fold colors  
$c^f_{\GXone}(u_{k_1})$ and $c^f_{\GXone}(u_{k_2})$ can have no overlaps, i.e., ${\rm dist}(c^f_{\GXone}(u_{k_1}),c^f_{\GXone}(u_{k_2}))=b$. 
More generally, a valid $a:b$ coloring of the EWCG $\GXone$ is such that given $w(u_{k_1}, u_{k_2})$, the minimum number of disjoint colors between $u_{k_1}$ and $u_{k_2}$ of $\GXone$ is
\begin{align}
\label{distance_between_vertices}
{\rm dist}(c^f_{\GXone}(u_{k_1}),c^f_{\GXone}(u_{k_2}))=\lceil w(u_{k_1}, u_{k_2}) \cdot b\rceil \ ,
\end{align}
meaning that if $w(u_{k_1}, u_{k_2}) \in \Big(\frac{b-(k+1)}{b}, \frac{b-k}{b}\Big]$ for $k\in \{0\} \cup [b-1]$, then vertices $u_{k_1}$ and $u_{k_2}$ are assigned $b-k$ distinct colors, and only if $w(u_{k_1}, u_{k_2})=0$ they are assigned exactly the same $b$ colors. 
We note that the number of different colors between two vertices of $\GXone$ changes as a function of the edge weight, as given in (\ref{distance_between_vertices}). 
The neighboring vertices in $G^w_{X_{1i}}$ have at least one different color, and the endpoints of edges with large weights have a higher number of disjoint colors. 
Clearly, this coloring scheme generalizes the notion of fractional chromatic number (Defn. \ref{fractional_chromatic_number} in App. \ref{preliminary}).

We next expand $\GXone$ into a $b$-tuple of graphs represented by $G^w_{X_1(S)}=\{G^w_{X_{1i}}:i\in S,\, |S|=b\}$, where 
$G^w_{X_{1i}}$, $i\in S$ is an $i$-th replica of $\GXone$. We jointly color the set of graphs $G^w_{X_1(S)}$ such that $c_{{G_{X_1(S)}}}(X_1(S))=\{c_{{G^w_{X_{1i}}}}(X_{1i}):i\in S,\, |S|=b\}$ and $w(u_{k_1},u_{k_2})=\frac{1}{b}\sum\limits_{i\in S} w_i(u_{k_1},u_{k_2})$ is split such that 
\begin{align}
\label{weighting_b_folded}
&w_i(u_{k_1},u_{k_2})\\
&=\min\left\{1,\, \max\Big\{b\cdot  w(u_{k_1},u_{k_2})-b\sum\limits_{i'=1}^{i-1} w_{i'}(u_{k_1},u_{k_2}),\, 0\Big\}\right\} \ \nonumber
\end{align}
denotes the weight between the vertices 
$u_{k_1}$ and $u_{k_2}$ of $G^w_{X_{1i}}$, $i\in S$, i.e., the $i$-th replica of $\GXone$. 
Note that (\ref{weighting_b_folded}) yields a sequence of monotone decreasing edge weights $w_i(u_{k_1},u_{k_2})$ for $i\in S$  
that jointly determine the traditional colorings for the set of graphs $G^w_{X_1(S)}$.  
In Fig. \ref{edge_weighted_graphs_btwo_2ndcoloringscheme}, we show a joint coloring for an example $|S|=2$-tuple EWCG. We will detail this example in Sect. \ref{sec:example} to indicate the achievable gains in compression.

We next explore the fundamental rate limits for distributed computing of $f(X_1,\,X_2)$, by exploiting the notions of characteristic graph entropy, and EWCGs, where we determine the weights according to (\ref{weighting_b_folded}), following the bipartite projection scheme. To that end, we next detail encoding and decoding of $\GXone$ (using the edge weights) for asymptotically lossless compression of $f({\bf X}_1^n,\,{\bf X}_2^n)$.

\subsection{An Achievable Coloring Scheme for Edge-Weighted Graphs}
\label{sec:achievable_coloring_weighted_graphs}

In this part, we detail the encoding and decoding principle of EWCGs for distributed computing of $f(X_1,\,X_2)$. 
We next describe the encoding of $b$-fold colors. 
Note that the computation of $f$ is lossless independent of the value of $b\in\mathbb{Z}^+$. 

\paragraph{Encoding}
\label{sec:encoding}
Given $G_f$, the encoding phase includes the projections of $G_f$ onto $\GXone$ and $\GXtwo$ by determining the corresponding edge weights using (\ref{weighting}) followed by their normalization. 
Each source then builds a $b$-tuple of characteristic graphs, namely 
$G^w_{X_1(S)}$ and $G^w_{X_2(S)}$, respectively, for $|S|=b$. 
The sources can then compress 
their weighted graphs 
asymptotically at rates $H^f_{\GXone}(X_1)$ and $H^f_{\GXone}(X_2)$, where we next give the conditional {\FCGE} of the EWCG $\GXone$.

\begin{figure*}[t!]
\centering
\includegraphics[width=\textwidth]{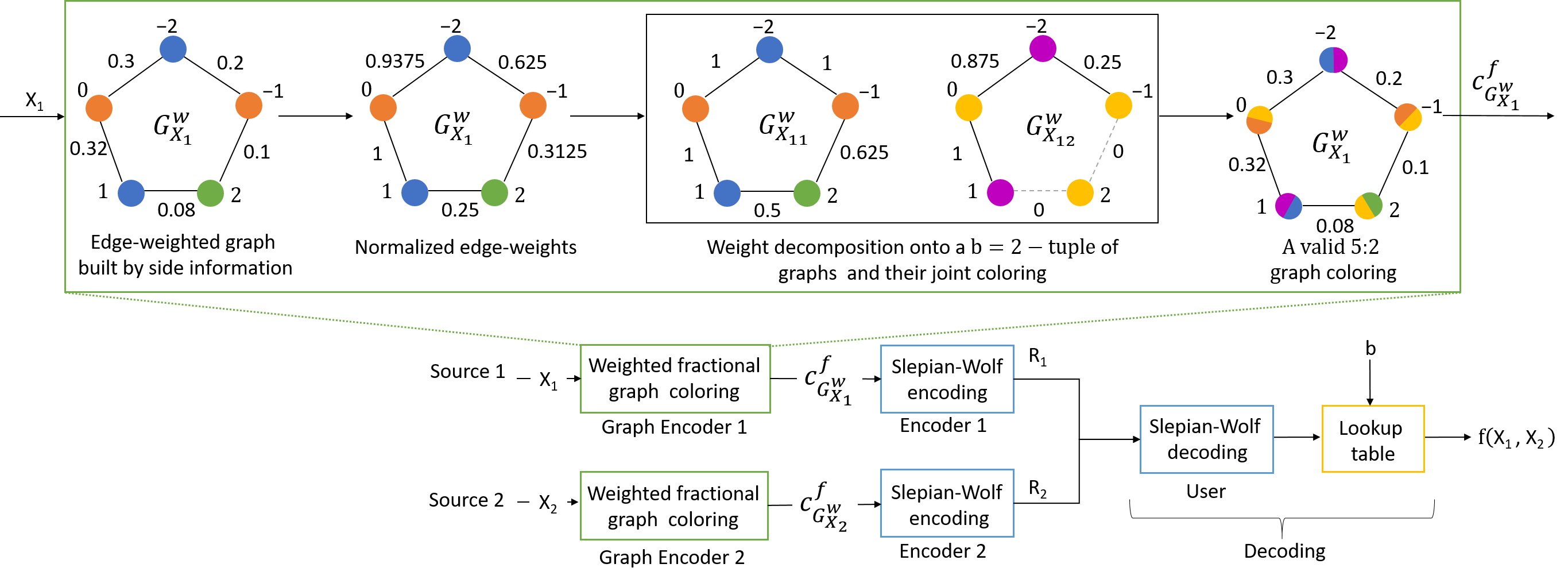}
\caption{Distributed computation of $f(X_1,\,X_2)$: an end-to-end multi-fold encoding and decoding scheme for EWCGs. 
The encoding phase consists of determining the EWCG tuples and their colorings, followed by Slepian-Wolf encoding on the $b$-fold colors. Decoding relies on recovering the $b$-fold colors using Slepian-Wolf decoding followed by recovering the outcomes using a look-up table. 
We note that $P_{X_1}=(0.2,0.15,0.32,0.24,0.09)$, and the edge weights are given in Example \ref{uniform_weighted_example}. 
In the bottom figure, the graph encoder for  
each source is independent -- output is a $b$-tuple color sequence -- with Slepian-Wolf encoding.  
In this example, the user uses $4$ fractional colors received ($b=2$ from each source per transmission) to reconstruct the function outcome.}
\label{edge_weighted_graphs_btwo_2ndcoloringscheme}
\end{figure*}

\begin{theo}
\label{fractional_graph_entropy_edge_weighted_graph}
The {\FCGE} of $\GXone$ is equal to 
\begin{multline}
\label{weight_compression_rate_asymptotic}
H^f_{\GXone}(X_1\, \vert\, X_2)=\lim\limits_{n\to\infty} \frac{1}{n} \inf\limits_{b} \frac{1}{b} \min\limits_{c^f_{\GXonen}} \{H(c^f_{\GXonen}({\bf X}_1))\,:\\
c^f_{\GXonen}({\bf X}_1) \mbox{ is a valid } a:b \mbox{ coloring of } \GXonen \, \vert\, {\bf X}_2\} \ ,
\end{multline}
where $c^f_{\GXonen}({\bf X}_1)$ is a fractional coloring variable for $\GXonen$ with an $a:b$ coloring of each vertex of $\GXonen$. 
\end{theo}

\begin{proof}
A proof sketch is given in App. \ref{app:fractional_graph_entropy_edge_weighted_graph}.
\end{proof}

\paragraph{Decoding} 
\label{sec:decoding}
For lossless decoding, the user needs to be instructed on the joint PMF $P_{X_1,X_2}$, the desired function $f$, $b$, and the look-up table for recovering $f({\bf X}_1^n,\,{\bf X}_2^n)$ using the received fractional colorings of the $b$-tuple of graphs from each source. 
The user first performs minimum-entropy decoding on its received information \cite{csiszar2011information}. 
Via Slepian-Wolf decoding, it achieves the random sequences $c^f_{\GXonen}({\bf X}_1)$ and $c^f_{\GXtwon}({\bf X}_2)$ that model the $b$-fold color tuples.  
The user then uses a look-up table to compute $f({\bf X}_1^n,\,{\bf X}_2^n)$.

To demonstrate the procedure for encoding and decoding of an EWCG $\GXone$, determining the edge weights in (\ref{weighting_b_folded}), and sending a pair of $b$-tuples of coloring sequences for recovery of $f({\bf X}_1^n,\,{\bf X}^n_2)$ by the user in an asymptotically lossless manner, we next detail an end-to-end distributed computing example with a $b=2$-fold coloring of $\GXone$, which is shown in Fig. \ref{edge_weighted_graphs_btwo_2ndcoloringscheme}.

\subsection{An Example toward Edge-Weighted Encoding-Decoding }
\label{sec:example}

We present an example to illustrate how to build an EWCG and how to encode and decode the coloring, to obtain the desired function outcomes. 
Through this example, we also contrast the performance of our scheme with that of traditional graph coloring that does not exploit the weight information.

\begin{ex}\label{uniform_weighted_example}{\bf An EWCG and its chromatic entropy.} 
The source variables $X_1$ and $X_2$ share a common alphabet such that $\mathcal{X}=\{-2,\,-1,\,0,\,1,\,2\}$. The ordered marginal PMFs are $X_1\sim {\bf p}_1
=(0.2,0.15,0.32,0.24,0.09)$ and $X_2\sim {\bf p}_2=(0.2,0.3,0.32,0.08,0.1)$, and  
$P_{X_1,X_2}$ is given as follows:
\begin{align}
\label{JointProbabilityTable_weighted}
P_{X_1,\,X_2}=\begin{bmatrix}
0.1 & 0.1 & 0 & 0 & 0\\
0.1 & 0 & 0 & 0 & 0.05\\
0 & 0.2 & 0.12 & 0 & 0\\
0 & 0 & 0.2 & 0.04 & 0\\
0 & 0 & 0 & 0.04 & 0.05
\end{bmatrix}.  
\end{align} 
We note that the entropy of $X_1$ satisfies $H(X_1)=h(0.2,0.15,0.32,0.24,0.09)=2.2078<H(X_{1,u})=2.32$, with $X_{1,u}\sim P_{X_{1,u}}$, where $P_{X_{1,u}}$ is uniform over $\mathcal{X}$.

{\bf Unweighted scenario.}
Without taking into account the edge weights, the minimum entropy coloring of $G_{X_1}$ is given as $H(c_{G_{X_1}})=h(0.44,0.47,0.09)=1.35$. 
The entropy with a $5:2$ fractional coloring with $\chi_f(G_{X_1})=2.5$ satisfies
$\frac{1}{2}H(c^f_{G_{X_1}})=\frac{1}{2}h(0.22,0.235,0.205,0.145,0.195)=1.15$. 
Similarly, for the second power 
$G^2_{{\bf X}_1}$, with an $8:1$ coloring, and a PMF \cite{feizi2014network} 
\begin{align}
c_{G^2_{{\bf X}_1}}\sim
(&0.176,0.188,0.018,0.176,\nonumber\\
&0.188,0.036,0.036,0.182)\ ,\nonumber
\end{align}
we get 
$\frac{1}{2}H(c_{G^2_{{\bf X}_1}})=1.34$. For a $13:2$ 
coloring, 
%
%
$\frac{1}{4}H(c^f_{G^2_{{\bf X}_1}})=0.91$. 
For $X_{1,u}$ uniform, it holds that  $\frac{1}{2}H(c^f_{G_{X_{1,u}}})=\frac{1}{2}\log 5=1.16$, and for ${\bf X}_{1,u}$ uniform, $\frac{1}{4}H(c^f_{G^2_{{\bf X}_{1,u}}})=\frac{1}{4}\log 13=0.92$.

\begin{figure*}[t!]
\centering
\includegraphics[width=\textwidth]{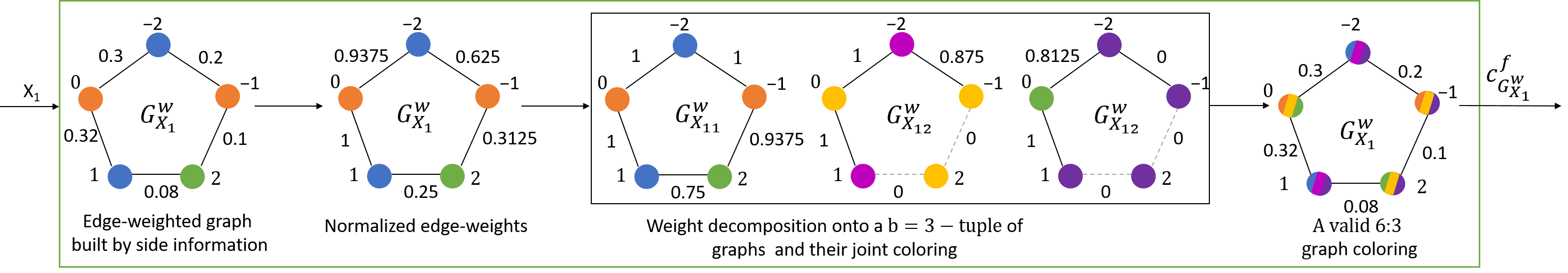}
\caption{A fractional coloring scheme for distributed computation of $f(X_1,\,X_2)$ with a $6:3$ coloring.}
\label{edge_weighted_graphs_bthree}
\end{figure*}
{\bf Weighted scenario.} 
We next take into account the edge weights. 
Using (\ref{weighting}), the edge weights are $w(-2,-1)=0.2$, $w(-2,0)=0.3$, $w(0,1)=0.32$, $w(1,2)=0.08$, and $w(-1,2)=0.1$. Note that for this specific example, $W\sim {\bf p}_2$.

We next decompose $\GXone$ into $b=2$ graphs, as shown in Fig. \ref{edge_weighted_graphs_btwo_2ndcoloringscheme} (top row).   
Normalizing the edge weights to set the maximum weight to be one, and then using (\ref{weighting_b_folded}), the weights are $w_1(-2,-1)=w_1(-2,0)=w_1(0,1)=1$, $w_1(1,2)=0.5$, and $w_1(-1,2)=0.625$ for  $G^w_{X_{11}}$, and 
$w_2(-2,-1)=0.25$, $w_2(-2,0)=0.875$, $w_2(0,1)=1$, and $w_2(1,2)=w_2(-1,2)=0$ for
$G^w_{X_{12}}$. 
This yields a valid $5:2$ coloring of $G^w_{X_1(S)}$ for $|S|=2$, 
as also shown in the top row.

Using the joint coloring information of $G^w_{X_1(S)}$, i.e., for $G^w_{X_{11}}$ and $G^w_{X_{12}}$, 
the color PMF for the $5:2$ fractional coloring of $\GXone$ for the set of ordered colors $\{c_1=Blue, c_2=Orange, c_3=Green, c_4=Purple, c_5=Yellow\}$ satisfies 
\begin{small}
\begin{align}
\label{marginal_colors_n_one_v2}
P_{c^f_{\GXone}}(c_1)&=\frac{1}{2}({\bf p}_1(-2)+{\bf p}_1(1))=\frac{0.44}{2}=0.22\ ,\nonumber\\ 
P_{c^f_{\GXone}}(c_2)&=\frac{1}{2}({\bf p}_1(-1)+{\bf p}_1(0))=\frac{0.47}{2}=0.235\ ,\nonumber\\ 
P_{c^f_{\GXone}}(c_3)&=\frac{1}{2}{\bf p}_1(2)=0.045\ ,\nonumber\\ 
P_{c^f_{\GXone}}(c_4)&=\frac{1}{2}({\bf p}_1(-2)+{\bf p}_1(1))=\frac{0.44}{2}=0.22\ ,\nonumber\\  P_{c^f_{\GXone}}(c_5)&=\frac{1}{2}({\bf p}_1(-1)+{\bf p}_1(0)+{\bf p}_1(2))=\frac{0.56}{2}=0.28\ ,
\end{align}
\end{small}
which yields from (\ref{weight_compression_rate_asymptotic}) that $\frac{1}{2}H(c^f_{\GXone})=1.08< \frac{1}{2}H(c^f_{G_{X_1}})=1.15<H(c_{G_{X_1}})=1.35$. 
Hence, for $b=2$, capturing the edge weights yields a saving of $\% 16$ over traditional coloring and does not offer enhancement over standard fractional coloring that does not capture the weights.

For the same example, with $b=3$ and with the inclusion of a sixth color, where $c_6=Violet$, we can achieve a $6:3$ coloring as shown in Fig. \ref{edge_weighted_graphs_bthree}, and the coloring PMF is 
\begin{small}
\begin{align}
\label{marginal_colors_n_one_bthree_v2}
P_{c^f_{\GXone}}(c_1)&=P_{c^f_{\GXone}}(c_4)=\frac{1}{3}({\bf p}_1(-2)+{\bf p}_1(1))=\frac{0.44}{3}\ ,\nonumber\\ 
P_{c^f_{\GXone}}(c_2)&=\frac{1}{3}({\bf p}_1(-1)+{\bf p}_1(0))=\frac{0.47}{3}\ ,\nonumber\\ 
P_{c^f_{\GXone}}(c_3)&=\frac{1}{3}({\bf p}_1(2)+{\bf p}_1(0))=\frac{0.41}{3}\ ,\nonumber\\ 
P_{c^f_{\GXone}}(c_5)&=\frac{1}{3}({\bf p}_1(-1)+{\bf p}_1(0)+{\bf p}_1(2))=\frac{0.56}{3}\ ,\nonumber\\
P_{c^f_{\GXone}}(c_6)&=
\frac{1}{3}(1-{\bf p}_1(0))=\frac{0.68}{3}\ .
\end{align}
\end{small}
Then, a valid $6:3$ coloring of $\GXone$ yields $\frac{1}{3}H(c^f_{\GXone})=0.85$, providing a saving of $\% 37$ over traditional coloring. Hence, a larger $b$ can capture the edge weights more accurately.

We next consider the second power graph $\GXonetwo$. 
We note that $\chi_f(\GXone)=2.5$, and 
$\chi_f(\GXonetwo)=\chi_f^2(\GXone)=6.25$. Hence, a $12:2$ coloring is not possible for $n=2$. We show a valid $13:2$  
coloring of $\GXonetwo$ in Fig. 
\ref{Example2_nonuniform_second_power_graph_2ndcoloringscheme}, given the ordered set  
$\{c_1=Blue,c_2=Yellow,c_3=Green,c_4=Orange,c_5=Purple,c_6=LightBlue,c_7=Brown,c_8=Violet,c_9=BrickRed,c_{10}=DarkGreen,c_{11}=Black,c_{12}=Gray,c_{13}=Navy\}$. Its 
coloring PMF can be derived from that for $\GXone$ and can be shown to satisfy 
\begin{small}
\begin{align}
P_{c^f_{\GXonetwo}}(c_m)&=\frac{2}{5}P_{c^f_{\GXone}}(c_1)=
0.088\ ,\quad m\in \{1, \, 11, \, 12\} \ , 
\nonumber\\
P_{c^f_{\GXonetwo}}(c_2)&=\frac{2}{5}P_{c^f_{\GXone}}(c_5)=0.112\ , 
\nonumber\\
P_{c^f_{\GXonetwo}}(c_3)&=\frac{1}{5}(P_{c^f_{\GXone}}(c_3)+P_{c^f_{\GXone}}(c_2))=0.056\ , 
\nonumber\\
P_{c^f_{\GXonetwo}}(c_4)&=\frac{2}{5}P_{c^f_{\GXone}}(c_2)=0.094\ , 
\nonumber\\
P_{c^f_{\GXonetwo}}(c_5)&=\frac{1}{5}(P_{c^f_{\GXone}}(c_4)+P_{c^f_{\GXone}}(c_3))=0.053\ , 
\nonumber\\
P_{c^f_{\GXonetwo}}(c_6)&=\frac{1}{5}(P_{c^f_{\GXone}}(c_4)+P_{c^f_{\GXone}}(c_5))=0.1\ , 
\nonumber\\
P_{c^f_{\GXonetwo}}(c_7)&=\frac{1}{5}(P_{c^f_{\GXone}}(c_1)+P_{c^f_{\GXone}}(c_2))=0.091\ , 
\nonumber\\
P_{c^f_{\GXonetwo}}(c_8)&=\frac{2}{5}P_{c^f_{\GXone}}(c_3)=0.018, 
\nonumber\\
P_{c^f_{\GXonetwo}}(c_9)&=\frac{1}{5}(P_{c^f_{\GXone}}(c_3)+P_{c^f_{\GXone}}(c_5))=0.065\ , 
\nonumber\\
P_{c^f_{\GXonetwo}}(c_{10})&=\frac{1}{5}(P_{c^f_{\GXone}}(c_2)+P_{c^f_{\GXone}}(c_5))=0.103\ , 
\nonumber\\
P_{c^f_{\GXonetwo}}(c_{13})&=\frac{1}{5}P_{c^f_{\GXone}}(c_4)=0.044 \ , 
\nonumber
\end{align}
\end{small}
which yields from (\ref{weight_compression_rate_asymptotic}) that $\frac{1}{4}H(c^f_{\GXonetwo}({\bf X}_1))=0.9=\frac{1}{4}H(c^f_{G^2_{{\bf X}_1}})=0.91<\frac{1}{2}H(c_{G^2_{{\bf X}_1}})=1.34$. 
Hence, capturing the edge weights yields a saving of $\% 32$ over traditional coloring, and does not have much gain over the fractional coloring approach that does not capture the weights. Increasing $b$ allows us to capture the edge weights more accurately.

\begin{figure}[h!]
\centering
\includegraphics[width=0.65\columnwidth]{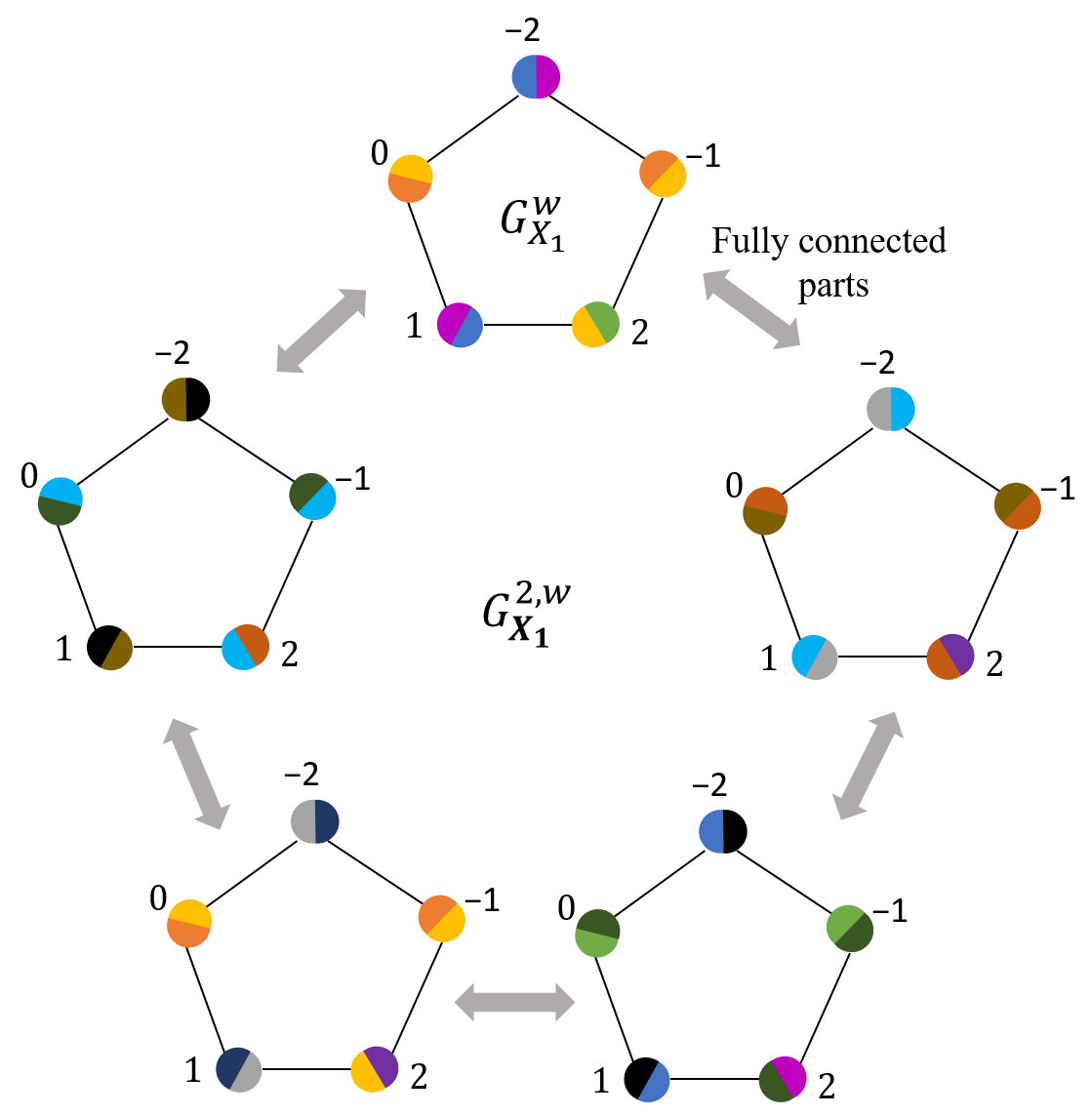}
\caption{A valid $13:2$ fractional coloring of $\GXonetwo$ for Example \ref{uniform_weighted_example}, where $\chi_f(\GXonetwo)=\chi_f^2(\GXone)=(2.5)^2=6.25$.}
\label{Example2_nonuniform_second_power_graph_2ndcoloringscheme}
\end{figure}

Similarly, we can determine the compression rate for general $n$. Exploiting \cite[Cor. 3.4.3]{scheinerman2011fractional}, $\chi(G_{{\bf X}_1}^n) \approx \chi_f^n(G_{X_1})$ as $n$ goes to infinity. Hence, we can derive the $n$-th power graph, $\GXonen$, along with its $a:b$ fractional coloring, $c^f_{\GXonen}({\bf X}_1)$. 
\end{ex}

From Example \ref{uniform_weighted_example}, as $b$ increases, we have a finer-grained quantization of the graph edge weights. 
As the skew of the edge weights increases, the efficiency in compressing the $b$-tuples of $\GXone$ increases (e.g., in Fig. \ref{edge_weighted_graphs_btwo_2ndcoloringscheme} some edges have relatively low weights, e.g., $w(1,2)=0.08$, and $w(-1,2)=0.1$, yielding a fewer number of total distinct colors between these two end vertices). 
As the value of $b$ increases, the edge weights will be captured with greater precision, leading to a more refined fractional coloring (more skewed) and a reduced total number of colors and smaller graph entropy 
$H^f_{\GXone}(X_1)$ given by (\ref{weight_compression_rate_asymptotic}).

When the {\em total bit budget for quantization and compression} is limited, there is a tradeoff between $b$ that determines the fold of coloring, and the complexity of encoding the characteristic graph. That is, the number of bits spent on quantizing the edge weights determines the attainable gains in compression.

\begin{spacing}{1}
\bibliographystyle{IEEEtran}
\bibliography{references}

\begin{thebibliography}{10}
\providecommand{\url}[1]{#1}
\csname url@samestyle\endcsname
\providecommand{\newblock}{\relax}
\providecommand{\bibinfo}[2]{#2}
\providecommand{\BIBentrySTDinterwordspacing}{\spaceskip=0pt\relax}
\providecommand{\BIBentryALTinterwordstretchfactor}{4}
\providecommand{\BIBentryALTinterwordspacing}{\spaceskip=\fontdimen2\font plus
\BIBentryALTinterwordstretchfactor\fontdimen3\font minus
  \fontdimen4\font\relax}
\providecommand{\BIBforeignlanguage}[2]{{%
\expandafter\ifx\csname l@#1\endcsname\relax
\typeout{** WARNING: IEEEtran.bst: No hyphenation pattern has been}%
\typeout{** loaded for the language `#1'. Using the pattern for}%
\typeout{** the default language instead.}%
\else
\language=\csname l@#1\endcsname
\fi
#2}}
\providecommand{\BIBdecl}{\relax}
\BIBdecl

\bibitem{soleymani2021analog}
M.~Soleymani, H.~Mahdavifar, and A.~S. Avestimehr, ``Analog lagrange coded
  computing,'' \emph{IEEE J. Sel. Areas Inf. Theory}, vol.~2, no.~1, pp.
  283--295, Feb. 2021.

\bibitem{wang2021price}
J.~Wang, Z.~Jia, and S.~A. Jafar, ``Price of precision in coded distributed
  matrix multiplication: A dimensional analysis,'' in \emph{Proc., IEEE Inf.
  Theory Wksh.}, Virtual Conference, Oct. 2021, pp. 1--6.

\bibitem{ho2006random}
T.~Ho, M.~M{\'e}dard, R.~Koetter, D.~R. Karger, M.~Effros, J.~Shi, and
  B.~Leong, ``A random linear network coding approach to multicast,''
  \emph{IEEE Trans. Inf. Theory}, vol.~52, no.~10, pp. 4413--4430, Sep. 2006.

\bibitem{behrouzi2020efficient}
A.~Behrouzi-Far and E.~Soljanin, ``Efficient replication for straggler
  mitigation in distributed computing,'' \emph{arXiv preprint
  arXiv:2006.02318}, Jun. 2020.

\bibitem{yao1979some}
A.~C.-C. Yao, ``Some complexity questions related to distributive computing
  (preliminary report),'' in \emph{Proc. ACM Symp. Theory of Computing},
  Atlanta, GA, Apr. 1979, pp. 209--213.

\bibitem{YuRavSoAve2018}
Q.~Yu, S.~Li, N.~Raviv, S.~M.~M. Kalan, M.~Soltanolkotabi, and S.~A.
  Avestimehr, ``Lagrange coded computing: Optimal design for resiliency,
  security, and privacy,'' in \emph{Proc., Int. Conf. Artificial Intelligence
  and Statistics}.\hskip 1em plus 0.5em minus 0.4em\relax Naha, Okinawa, Japan:
  PMLR, Apr. 2019, pp. 1215--1225.

\bibitem{reisizadeh2019coded}
A.~Reisizadeh, S.~Prakash, R.~Pedarsani, and A.~S. Avestimehr, ``Coded
  computation over heterogeneous clusters,'' \emph{IEEE Trans. Inf. Theory},
  vol.~65, no.~7, pp. 4227--4242, Mar. 2019.

\bibitem{prakash2020coded}
S.~Prakash, A.~Reisizadeh, R.~Pedarsani, and A.~S. Avestimehr, ``Coded
  computing for distributed graph analytics,'' \emph{IEEE Trans. Inf. Theory},
  vol.~66, no.~10, pp. 6534--6554, Jun. 2020.

\bibitem{jia2021capacity}
Z.~Jia and S.~A. Jafar, ``On the capacity of secure distributed batch matrix
  multiplication,'' \emph{IEEE Trans. Inf. Theory}, vol.~67, no.~11, pp.
  7420--7437, Sep. 2021.

\bibitem{wan2022cache}
K.~Wan, H.~Sun, M.~Ji, D.~Tuninetti, and G.~Caire, ``Cache-aided matrix
  multiplication retrieval,'' \emph{IEEE Trans. Inf. Theory}, Mar. 2022.

\bibitem{chang2018capacity}
W.-T. Chang and R.~Tandon, ``On the capacity of secure distributed matrix
  multiplication,'' in \emph{Proc., IEEE Global Commun. Conf.}, Abu Dhabi, UAE,
  Dec. 2018, pp. 1--6.

\bibitem{chen2021gcsa}
Z.~Chen, Z.~Jia, Z.~Wang, and S.~A. Jafar, ``{GCSA} codes with noise alignment
  for secure coded multi-party batch matrix multiplication,'' \emph{IEEE J.
  Sel. Areas Inf. Theory}, vol.~2, no.~1, pp. 306--316, Jan. 2021.

\bibitem{jia2021cross}
Z.~Jia and S.~A. Jafar, ``Cross subspace alignment codes for coded distributed
  batch computation,'' \emph{IEEE Trans. Inf. Theory}, vol.~67, no.~5, pp.
  2821--2846, Mar. 2021.

\bibitem{li2021flexible}
W.~Li, Z.~Chen, Z.~Wang, S.~A. Jafar, and H.~Jafarkhani, ``Flexible
  constructions for distributed matrix multiplication,'' in \emph{Proc., IEEE
  Int. Symp. Inf. Theory}, Virtual Conference, Jul. 2021, pp. 1576--1581.

\bibitem{wan2020cache}
K.~Wan, H.~Sun, M.~Ji, D.~Tuninetti, and G.~Caire, ``Cache-aided general linear
  function retrieval,'' \emph{Entropy}, vol.~23, no.~1, p.~25, Dec. 2020.

\bibitem{khalesi2022multi}
A.~Khalesi and P.~Elia, ``Multi-user linearly-separable distributed
  computing,'' \emph{arXiv preprint arXiv:2206.11119}, Jun. 2022.

\bibitem{wan2021distributed}
K.~Wan, H.~Sun, M.~Ji, and G.~Caire, ``Distributed linearly separable
  computation,'' \emph{IEEE Trans. Inf. Theory}, vol.~68, no.~2, pp.
  1259--1278, Nov. 2021.

\bibitem{wan2022secure}
------, ``On secure distributed linearly separable computation,'' \emph{IEEE J.
  Sel. Areas Commun.}, vol.~40, no.~3, pp. 912--926, Jan. 2022.

\bibitem{SlepWolf1973}
D.~Slepian and J.~K. Wolf, ``Noiseless coding of correlated information
  sources,'' \emph{IEEE Trans. Inf. Theory}, vol.~19, no.~4, pp. 471--480, Jul.
  1973.

\bibitem{WynZiv1976}
A.~Wyner and J.~Ziv, ``The rate-distortion function for source coding with side
  information at the decoder,'' \emph{IEEE Trans. Inf. Theoy}, vol.~22, no.~1,
  pp. 1--10, Jan. 1976.

\bibitem{yamamoto1982wyner}
H.~Yamamoto, ``{Wyner-Ziv} theory for a general function of the correlated
  sources,'' \emph{IEEE Trans. Inf. Theory}, vol.~28, no.~5, pp. 803--7, Sep.
  1982.

\bibitem{OR01}
A.~Orlitsky and J.~R. Roche, ``Coding for computing,'' \emph{IEEE Trans. Inf.
  Theory}, vol.~47, no.~3, p. 903–917, Mar. 2001.

\bibitem{Kor73}
J.~K\"orner, ``Coding of an information source having ambiguous alphabet and
  the entropy of graphs,'' in \emph{Proc., 6th Prague Conf. Inf. Theory},
  Prague, Czech Republic, Sep. 1973, pp. 411--425.

\bibitem{AO96}
N.~Alon and A.~Orlitsky, ``Source coding and graph entropies,'' \emph{IEEE
  Trans. Inf. Theory}, vol.~42, no.~5, pp. 1329--1339, Sep. 1996.

\bibitem{feizi2014network}
S.~Feizi and M.~M{\'e}dard, ``On network functional compression,'' \emph{IEEE
  Trans. Inf. Theory}, vol.~60, no.~9, pp. 5387--5401, Sep. 2014.

\bibitem{doshi2007source}
V.~Doshi, D.~Shah, and M.~M{\'e}dard, ``Source coding with distortion through
  graph coloring,'' in \emph{Proc., IEEE Int. Symp. Inf. Theory}, Nice, France,
  Jun. 2007, pp. 1501--1505.

\bibitem{basu2022hypergraph}
S.~Basu, D.~Seo, and L.~R. Varshney, ``Hypergraph-based source codes for
  function computation under maximal distortion,'' \emph{arXiv preprint
  arXiv:2204.02586}, Apr. 2022.

\bibitem{basu2020functional}
------, ``Functional epsilon entropy,'' in \emph{Proc., Data Compression
  Conf.}, Virtual Conference, Mar. 2020, pp. 332--341.

\bibitem{Malak2022hyperbin}
D.~Malak and M.~M\'edard, ``A distributed computationally aware quantizer
  design via hyper binning,'' \emph{IEEE Trans. Signal Proces.,}, Jan. 2023.

\bibitem{cover1975proof}
T.~Cover, ``A proof of the data compression theorem of {Slepian and Wolf} for
  ergodic sources (corresp.),'' \emph{IEEE Trans. Inf. Theory}, vol.~21, no.~2,
  pp. 226--228, Mar. 1975.

\bibitem{malak2022fractional}
D.~Malak, ``Fractional graph coloring for functional compression with side
  information,'' in \emph{Proc., IEEE Inf. Theory Wksh.}, Mumbai, India, Nov.
  2022.

\bibitem{korner1979encode}
J.~K{\"o}rner and K.~Marton, ``How to encode the modulo-two sum of binary
  sources (corresp.),'' \emph{IEEE Trans. Inf. Theory}, vol.~25, no.~2, pp.
  219--221, Mar. 1979.

\bibitem{adikari2022two}
T.~Adikari and S.~Draper, ``Two-terminal source coding with common sum
  reconstruction,'' in \emph{Proc., IEEE Int. Symp. Inf. Theory}, Espoo,
  Finland, Jun. 2022, pp. 1420--1424.

\bibitem{suresh2017distributed}
A.~T. Suresh, X.~Y. Felix, S.~Kumar, and H.~B. McMahan, ``Distributed mean
  estimation with limited communication,'' in \emph{Proc., Int. Conf. Machine
  Learning}.\hskip 1em plus 0.5em minus 0.4em\relax Sydney, Australia: PMLR,
  Jul. 2017, pp. 3329--3337.

\bibitem{han1987dichotomy}
T.~Han and K.~Kobayashi, ``{A dichotomy of functions F (X, Y) of correlated
  sources (X, Y)},'' \emph{IEEE Trans. Inf. Theory}, vol.~33, no.~1, pp.
  69--76, Jan. 1987.

\bibitem{witsenhausen1976zero}
H.~Witsenhausen, ``The zero-error side information problem and chromatic
  numbers (corresp.),'' \emph{IEEE Trans. Inf. Theory}, vol.~22, no.~5, pp.
  592--593, Sep. 1976.

\bibitem{ahlswede1979coloring}
R.~Ahlswede, ``Coloring hypergraphs: A new approach to multi-user source
  coding,'' \emph{J. Comb.}, vol.~4, no.~1, pp. 76--115, 1979.

\bibitem{korner1998zero}
J.~K{\"o}rner and A.~Orlitsky, ``Zero-error information theory,'' \emph{IEEE
  Trans. Inf. Theory}, vol.~44, no.~6, pp. 2207--2229, Oct. 1998.

\bibitem{cover1980multiple}
T.~Cover, A.~E. Gamal, and M.~Salehi, ``Multiple access channels with
  arbitrarily correlated sources,'' \emph{IEEE Trans. Inf. Theory}, vol.~26,
  no.~6, pp. 648--657, Nov. 1980.

\bibitem{nazer2007computation}
B.~Nazer and M.~Gastpar, ``Computation over multiple-access channels,''
  \emph{IEEE Trans. Inf. Theory}, vol.~53, no.~10, pp. 3498--3516, Sep. 2007.

\bibitem{nazer2007lattice}
------, ``Lattice coding increases multicast rates for gaussian multiple-access
  networks,'' in \emph{Proc., Allerton Conf.}, Monticello, IL, Sep. 2007.

\bibitem{malak2023Structured}
D.~Malak, ``Distributed computing of functions of structured sources with
  helper side information,'' in \emph{Proc., IEEE Int. Wksh. Signal Proces.
  Advances in Wireless Commun}, Shanghai, China, Sep. 2023.

\bibitem{zhou2007bipartite}
T.~Zhou, J.~Ren, M.~Medo, and Y.-C. Zhang, ``Bipartite network projection and
  personal recommendation,'' \emph{Physical Review E}, vol.~76, no.~4, p.
  046115, Oct. 2007.

\bibitem{scheinerman2011fractional}
E.~R. Scheinerman and D.~H. Ullman, \emph{Fractional graph theory: a rational
  approach to the theory of graphs}.\hskip 1em plus 0.5em minus 0.4em\relax
  Courier Corporation, 2011.

\bibitem{chartrand2019chromatic}
G.~Chartrand and P.~Zhang, \emph{Chromatic Graph Theory}.\hskip 1em plus 0.5em
  minus 0.4em\relax CRC press, 2019.

\bibitem{shannon1956zero}
C.~Shannon, ``The zero error capacity of a noisy channel,'' \emph{IRE
  Transactions on Information Theory}, vol.~2, no.~3, pp. 8--19, Sep. 1956.

\bibitem{toivonen2011compression}
H.~Toivonen, F.~Zhou, A.~Hartikainen, and A.~Hinkka, ``Compression of weighted
  graphs,'' in \emph{Proc., ACM SIGKDD Int. Conf. Knowledge Discovery and Data
  Mining}, San Diego, CA, Aug. 2011, pp. 965--973.

\bibitem{csiszar2011information}
I.~Csisz{\'a}r and J.~K{\"o}rner, \emph{Information theory: coding theorems for
  discrete memoryless systems}.\hskip 1em plus 0.5em minus 0.4em\relax
  Cambridge University Press, Jun. 2011.

\bibitem{cover2012elements}
T.~M. Cover and J.~A. Thomas, \emph{Elements of Information Theory}.\hskip 1em
  plus 0.5em minus 0.4em\relax John Wiley \& Sons, 2012.

\bibitem{te1978nonnegative}
T.~S. Han, ``Nonnegative entropy measures of multivariate symmetric
  correlations,'' \emph{Information and Control}, vol.~36, pp. 133--156, Feb.
  1978.

\end{thebibliography}
\end{spacing}

\begin{appendices}
\section{Technical Preliminary}
\label{preliminary}
We consider a distributed  communication model with two sources and a user, where each source represents a \emph{projection of the bipartite graph} that describes the joint distribution of the sources. Hence, the sources have \emph{partial access to distributed source information}. 
The sources hold random variables $X_1$ and $X_2$, respectively. 
We assume that $X_1$ and $X_2$ model two statistically dependent i.i.d. finite alphabet source sequences with discrete alphabets $\mathcal{X}_1$ and $\mathcal{X}_2$, respectively, and they are jointly distributed according to $P_{X_1,X_2}$. Each source encodes its sequences independently via building an EWCG, which we detail in the following part, in App. \ref{sec:CharacteristicGraphs}. 

A source builds a characteristic graph and sends the coloring information (a $b$-fold coloring capturing the edge weights of the graph) to the user that performs {\em minimum-entropy decoding} on the received information. The user then uses a look-up table to compute the function by using the jointly distributed received color tuples. The user, exploiting the edge weights, computes a function $f(X_1,\,X_2)$ of $X_1$ and $X_2$ in an asymptotically lossless manner.  
To that end, our goal is to characterize an achievable rate region for this asymptotically lossless distributed computation problem.

To understand the fundamental limits of distributed computation, we next provide a primer on characteristic graphs, their traditional coloring and fractional coloring, and graph entropy.

\subsection{Source Characteristic Graphs and Their Vertex Colorings} 
\label{sec:CharacteristicGraphs}

Source one -- who does not have access to the outcomes of source two -- builds a characteristic graph $G_{X_1}=(V_{X_1},E_{X_1})$ for computing $f(X_1,\,X_2)$ to distinguish its outcomes that yield a different output for any value of $X_2$. Note that $V_{X_1}=\mathcal{X}_1$, and $E_{X_1}$ is determined as follows. Given two vertices $u_{k_1}, u_{k_2} \in V_{X_1}$ in $G_{X_1}$ such that $k_1\neq k_2$, if $\exists$ at least one vertex $v_{l}\in \mathcal{X}_2$ in $G_{X_2}$ such that $P_{X_1,X_2}(u_{k_1},v_{l})P_{X_1,X_2}(u_{k_2},v_{l})>0$ and the function satisfies $f(u_{k_1},v_{l})\neq f(u_{k_2},v_{l})$, then $(u_{k_1},u_{k_2})\in E_{X_1}$. Otherwise, $(u_{k_1},u_{k_2})\notin E_{X_1}$. Similarly, we can build $G_{X_2}$.

We let $c_{G_{X_1}}(X_1)$ be a valid vertex coloring of $G_{X_1}$, where a valid coloring is such that any two vertices of $G_{X_1}$ that share an edge are assigned distinct colors, i.e., edges of $G_{X_1}$ have unit weights. 
The joint PMF of $c_{G_{X_1}}$ and $c_{G_{X_2}}$ satisfies 
\begin{multline}
\label{color_class_distribution}
P_{c_{G_{X_1}},\,c_{G_{X_2}}}(c_{G_{X_1}}(u_{k}),\,c_{G_{X_2}}(v_{l}))\\
=
\sum\limits_{(u_k,v_l)\in \mathcal{J}(k,l)} P_{X_1,X_2}(x_1,x_2)\ ,
\end{multline}
where the sum over the joint coloring class 
$\mathcal{J}(k,l)=\{(u_{k},v_{l}),\,(x_1,x_2):c_{G_{X_1}}(u_k)=c_{G_{X_1}}(x_1),\,c_{G_{X_2}}(v_l)=c_{G_{X_2}}(x_2)\}$ for any valid $k$ and $l$, i.e., the collection of points $(u_{k},v_{l})$ whose coordinates have the same color. 
We assume that $P_{X_1,X_2}(x_1,x_2)>0$ for all $(x_1,x_2)$. Under this condition, we infer from \cite[Theorem 56]{feizi2014network} that maximal independent sets\footnote{
A maximal independent set (MIS) is an independent set that is not a subset of any other independent set.} (MISs) of $G_{X_m}$ are some non-overlapping fully-connected sets, and hence, for any $f(X_1,\,X_2)$, the minimum entropy coloring  
can be achieved in polynomial time by assigning different colors to the different MISs 
of $G_{X_m}$.

The condition $P_{X_1,X_2}(x_1,x_2)>0$ for all $(x_1,x_2)$ ensures for two points $(u_{k_1},v_{l_1})$ and $(u_{k_2},v_{l_2})$ that are in the same joint coloring class, they are fully-connected, i.e., the coloring connectivity condition (CCC), a necessary and sufficient condition for any achievable coding model that relies on colorings, is satisfied, and $f(u_{k_1},v_{l_1})=f(u_{k_2},v_{l_2})$.  
From \cite[Lemma 27]{feizi2014network}, for any two points $(u_{k_1},v_{l_1})$ and $(u_{k_2},v_{l_2})$ that are in the same joint coloring class, their function outcomes are the same if and only if the joint coloring class satisfies the CCC. 
We note in this case that the theorem of Slepian and Wolf can be applied to the joint PMF of colors given by (\ref{color_class_distribution}) to achieve distributed lossless computing in the asymptotic regime.

We also note that (\ref{color_class_distribution}) can be generalized to determine the colorings of the $n$-th power graphs $G_{{\bf X}_m}^n$ and their joint coloring classes $\mathcal{J}({\bf k}^n,{\bf l}^n)$ where ${\bf k}^n=k_1,k_2,\dots,k_n$ and ${\bf l}^n=l_1,l_2,\dots,l_n$. 
To capture the fundamental limits of asymptotically lossless compression for computation of the sequence of function outcomes $f({\bf X}_1^n,\,{\bf X}_2^n)=f(X_{11},\,X_{21}),f(X_{12},\,X_{22}),\dots,f(X_{1n},\,X_{2n})$, we similarly build the $n$-th power of $G_{X_1}$, i.e., $G_{{\bf X}_1}^n$. We note that $G_{{\bf X}_1}^n=(V_{X_1}^n,E_{X_1}^n)$ is an OR graph such that  $V_{X_1}^n=\mathcal{X}_1^n$ and if $(u_{k_1,i},u_{k_2,i})\in E_{X_1}$ for some coordinate $i\in [n]$, then $({\bf u}^n_{k_1},{\bf u}^n_{k_2})\in E_{X_1}^n$. 
The entropy of the characteristic graph $G_{X_1}$ is given by \cite{Kor73} 
\begin{align}
\label{chromatic_vs_characteristic_marginal}
H_{G_{X_1}}(X_1)=\lim_{n\to\infty} \min\limits_{c_{G_{{\bf X}_1}^n}}\frac{1}{n} H(c_{G_{{\bf X}_1}^n})\ , 
\end{align}
where the minimization is over the set of all valid colorings $c_{G_{{\bf X}_1}^n}({\bf X}_1)$ of $G_{{\bf X}_1}^n$. 
Similarly, conditional graph entropy \cite{OR01} and joint graph entropy satisfy the following relations:
\begin{align}
\label{chromatic_vs_characteristic}
&H_{G_{X_1}}(X_1\, \vert X_2)=\lim\limits_{n\to\infty}\,\, \min\limits_{c_{G^n_{{\bf X}_1}},\, c_{G^n_{{\bf X}_2}}} \frac{1}{n} H(c_{G^n_{{\bf X}_1}} \vert  c_{G^n_{{\bf X}_2}}) \ ,  \\  
&H_{G_{X_1},G_{X_2}}(X_1,\,X_2)=\lim\limits_{n\to\infty} \,\,\min\limits_{c_{G^n_{{\bf X}_1}},\, c_{G^n_{{\bf X}_2}}} \frac{1}{n} H(c_{G^n_{{\bf X}_1}} \, , c_{G^n_{{\bf X}_2}})\ .\nonumber
\end{align}

Exploiting (\ref{chromatic_vs_characteristic_marginal}) and (\ref{chromatic_vs_characteristic}), the rate region 
for distributed computation of $f(X_1,\,X_2)$ is given by the set of rates \cite{feizi2014network}:
\begin{align}
R_1&\geq H_{G_{X_1}}(X_1\, \vert X_2), \nonumber\\ 
R_2&\geq H_{G_{X_2}}(X_2\, \vert X_1), \nonumber\\ R_1+R_2&\geq H_{G_{X_1},G_{X_2}}(X_1,\, X_2)\ ,
\end{align}
noting that the chain rule of conditional entropy is not satisfied  
due to the Schur-concavity of minimum. 
From data processing, the rate region to compute $(X_1,\,X_2)$ in an asymptotically lossless manner, given by the coding theorem of Slepian-Wolf in \cite{SlepWolf1973}, is encompassed by the rate region for computing $f(X_1,\, X_2)$, given by \cite{feizi2014network}.

\subsection{Fractional Coloring of Characteristic Graphs}
\label{sec:FractionalCharacteristicGraphs}

Fractional graph coloring is a natural extension of traditional coloring such that in fractional coloring, each vertex is assigned a set of colors (versus one color only), and the adjacent vertices have disjoint sets of colors.

\begin{defi}
\label{valid_ab_coloring}
({\bf Scheinerman and Ullman~\cite{scheinerman2011fractional}.})
A valid $b$-fold coloring of $G=(V,E)$ is an assignment of sets of size $b$ to vertices $V$ such that adjacent vertices receive disjoint sets of colors. A valid $a:b$ coloring is a valid $b$-fold coloring out of $a$ available colors in total. 
\end{defi}

The $b$-fold chromatic number, $\chi_b(G)$, of graph $G=(V,E)$ represents the least $a$ such that an $a:b$ coloring exists.

\begin{defi}\label{fractional_chromatic_number}
({\bf Fractional chromatic number \cite{scheinerman2011fractional}.}) The {\FCN} is defined as
\begin{align}
\label{fractional_graph_coloring}
\chi_f(G):=\liminf\limits_{b\to\infty}\left\{\frac{\chi_b(G)}{b}\right\}=\inf\limits_{b} \frac{\chi_b(G)}{b}\ , 
\end{align}
where the existence of this limit follows from the sub-additivity of $b$-fold colorings, and the sub-additivity lemma. 
\end{defi}

The {\FCN} $\chi _{f}(G)$ can be obtained as a solution of the following linear program \cite{scheinerman2011fractional}:  
\begin{align}
\label{chi_f_optimal}
\chi _{f}(G) = \min\limits_{\forall x\in V} \left\{ \sum _{I\in {\mathcal {I}}(G)} x_{I} : \sum_{I\in\mathcal{I}(G,x)} x_I \geq 1 , \,\, x_I\geq 0 \right\}\ ,    
\end{align}
where $\mathcal{I}(G)$ is the set of all independent sets of $G$, and $\mathcal{I}(G,x)$ is the set of all  $\mathcal{I}(G)$ which include vertex $x\in V$.

\subsection{Joint Coloring of a $b$-Tuple of Characteristic Graphs}
\label{joint_coloring_b_tuple}

We denote a $b$-tuple of the characteristic graph $\Graph$ by $G_{X_1(S)}=\{G_{X_{1i}}:i\in S,\, |S|=b\}$, where each element $G_{X_{1i}}$, $i\in S$ is a replica of $\Graph$. We jointly color $G_{X_1(S)}$ such that $c_{{G_{X_1(S)}}}(X_1(S))=\{c_{{G_{X_{1i}}}}(X_{1i}):i\in S,\, |S|=b\}$. Using valid traditional colorings with $a$ colors in total across disjoint $|S|=b$ graphs in $S$, we denote by $c_{{G_{X_1(S)}}}(X_1(S))$ a valid traditional coloring of $G_{X_1(S)}$, and by $c^f_{{G_{X_1}}}(X_1)$ a valid fractional coloring that provides an $a:b$ coloring of $\Graph$. 
The entropy of $c_{{G_{X_1(S)}}}(X_1(S))$ is
\begin{align}
\label{entropy_subset_full_coloring_vs_fractional_coloring}
    H(c_{{G_{X_1(S)}}}(X_1(S))\vert {\bf X}_2) = H(c^f_{{G_{X_1}}}(X_1)\vert {\bf X}_2) \ .
\end{align}
We denote the collection of fractional chromatic entropies over the set of all valid $a:b$ colorings of $G_{X_1}$ given $X_2$ by 
\begin{align}
\label{set_of_valid_factional_colorings}
\mathcal{H}^{\chi_f}(\coloringf)=\{H(\coloringf):\, \coloringf \mbox{ is a} \nonumber\\
\mbox{valid a:b coloring of } G_{X_1} \vert\, X_2\} \ .
\end{align}

The minimum entropy of a fractional coloring can be found by minimizing across all valid $a:b$ colorings of $\Graph$. 
We next state a characterization of the fractional graph entropy using the notion of fractional chromatic entropy \cite{malak2022fractional}.

\begin{prop}
\label{FCE_characteristic_graph_n_limit}
({\bf Fractional graph entropy \cite{malak2022fractional}.}) The {\FCGE} of a graph $G_{X_1}$ is given as
\begin{multline}
\label{FCGE}
\HGfrac=\lim\limits_{n\to \infty}\frac{1}{n} \inf\limits_{b} \frac{1}{b} \min\nolimits_{\coloringpowernxf}\{H(\coloringpowernf):\, \\ 
\coloringpowernf \mbox{ is a valid a:b coloring of } \nPowerGraph \vert\, {\bf X}_2^n\} \ ,   
\end{multline}
where $\coloringpowernf$ is a fractional coloring variable that assigns $b$ colors to each vertex of $\nPowerGraph$ out of $a\ge b$ available colors.
\end{prop}

\begin{proof}
Given a collection of random variables, ${\bf Z}_1^n=(Z_1,Z_2,\dots, Z_n)$, and every $S\subseteq [n]$, denote by $Z(S)=\{Z_i:\,i\in S\}$. From \cite[Ch. 16.5]{cover2012elements} the average entropy in bits per symbol of a randomly drawn $b$-element subset $Z(S)$ of ${\bf Z}_1^n$ is 
\begin{align}
\label{Han_joint_entropy}
\frac{1}{{n\choose b}} \sum\limits_{S:|S|=b} \frac{H(Z(S))}{b} \ , 
\end{align}
which decreases monotonically in the size of the subset \cite{te1978nonnegative}.    
Using (\ref{Han_joint_entropy}), the {\FCE} of $G_{X_1}$ equals  
\begin{align}
\label{chromatic_fractional}
\HGchifrac=\inf\limits_{b} \frac{1}{b} \min_{\coloringxf} \mathcal{H}^{\chi_f}(\coloringf)\ ,
\end{align} 
where $\chi_f$ is the fractional chromatic number of $\Graph$. 
We can observe from (\ref{chromatic_fractional}) that as $b$ increases, the rate of functional compression via fractional coloring decreases.  
Exploiting K\"orner's result \cite{Kor73}, the {\FCGE} satisfies
\begin{align}
\label{chromatic_vs_characteristic_fractional}
\HGfrac=\lim\limits_{n\to \infty}\frac{1}{n}\HGchipowernfrac \ ,
\end{align}
where $\chi_f(\nPowerGraph)$ is the fractional chromatic number of $\nPowerGraph$.

Using (\ref{chromatic_fractional}) and (\ref{chromatic_vs_characteristic_fractional}) we can derive the {\em \FCGE}, which is a natural generalization of the {\em conditional graph entropy} given in (\ref{chromatic_vs_characteristic}).  
\end{proof}

For further technical details and proofs, as well as the coding gains achieved by coloring, we refer the reader to \cite{OR01,Kor73,AO96,feizi2014network}.

\section{Proof of Theorem \ref{fractional_graph_entropy_edge_weighted_graph}}
\label{app:fractional_graph_entropy_edge_weighted_graph}
The proof follows along the same lines as Prop.~\ref{FCE_characteristic_graph_n_limit}. 
Using (\ref{Han_joint_entropy}), the {\FCE} of $\GXone$ is given as 
\begin{align}
\label{chromatic_fractional_weighted}
H^{\chi_f}_{\GXone}(X_1\, \vert\, X_2)=\inf\limits_{b} \frac{1}{b} \min_{c^f_{\GXone}} \mathcal{H}^{\chi_f}(c^f_{\GXone}(X_1))\ ,
\end{align} 
where $\chi_f$ is the fractional chromatic number of $\GXone$, i.e., $\chi_f(\GXone)$, and $\mathcal{H}^{\chi_f}(c^f_{\GXone}(X_1))$ follows from (\ref{set_of_valid_factional_colorings}) by substituting $G_{X_1}$ with $\GXone$. 
We can observe from (\ref{chromatic_fractional_weighted}) that as $b$ increases, the rate of functional compression via fractional coloring decreases.  
Exploiting K\"orner's result \cite{Kor73}, the {\FCGE} satisfies
\begin{align}
\label{chromatic_vs_characteristic_fractional_weighted}
H^f_{\GXonen}(X_1\, \vert\, X_2)=\lim\limits_{n\to \infty}\frac{1}{n}H^{\chi_f}_{\GXonen}({\bf X_1}\vert {\bf X_2}) \ ,
\end{align}
where $\chi_f(\GXonen)$ is the fractional chromatic number of $\GXonen$.

Using (\ref{chromatic_fractional_weighted}) and (\ref{chromatic_vs_characteristic_fractional_weighted}) we can derive the {\em \FCGE} for the edge-weighted graph $\GXone$, which is given as 
\begin{multline}
H^f_{\GXone}(X_1\, \vert\, X_2)=\lim\limits_{n\to\infty} \frac{1}{n} \inf\limits_{b} \frac{1}{b} \min\limits_{c^f_{\GXonen}} \{H(c^f_{\GXonen}({\bf X}_1))\,:\nonumber\\
c^f_{\GXonen}({\bf X}_1) \mbox{ is a valid } a:b \mbox{ coloring of } \GXonen \, \vert\, {\bf X}_2\} \ ,
\end{multline}
which is a natural generalization of the {\em conditional graph entropy} given in (\ref{chromatic_vs_characteristic}) and the {\em \FCGE} in (\ref{FCGE}).  

\end{appendices}
\end{document}